\documentclass[12pt]{article}
\usepackage{amsmath,amssymb}

\newtheorem{theorem}{Theorem}

\newtheorem{lemma}{Lemma}

\newtheorem{definition}{Definition}
\newtheorem{remark}{Remark}
\newtheorem{example}{Example}
\newtheorem{algorithm}{Algorithm}
\newenvironment{proof}{\noindent Proof:}{$\Box$}

\newcommand{\N}{{\mathbb N}}

\newcommand{\R}{{\mathbb R}}
\newcommand{\C}{{\mathbb C}}
\newcommand{\Z}{{\mathbb Z}}
\newcommand{\ord}{\mbox{{\rm ord}}}

\newcommand{\Ann}{\mbox{{\rm Ann}}}

\newcommand{\one}{{\bf 1}}

\newcommand{\Ssc}{{\mathcal S}}
\newcommand{\Lsc}{{\mathcal L}}
\newcommand{\Nsc}{{\mathcal N}}
\newcommand{\Msc}{{\mathcal M}}
\newcommand{\Psc}{{\mathcal P}}

\title{An algorithm to compute the differential equations for  
the logarithm of a polynomial}
\author{Toshinori Oaku\\ 
Department of Mathematics, 
Tokyo Woman's Christian University}
\date{January 12, 2012}

\begin{document}
\maketitle

\begin{abstract}
We present an algorithm to compute the annihilator of 
(i.e., the linear differential equations for) the 
multi-valued analytic function $f^\lambda(\log f)^m$ in the Weyl algebra  
$D_n$ for a given non-constant polynomial $f$, 
a non-negative integer $m$, and a complex number $\lambda$. 
This algorithm essentially consists in the differentiation with respect to 
$s$ of the annihilator of $f^s$ in the ring $D_n[s]$ and 
ideal quotient computation in $D_n$. 
The obtained differential equations constitute what is called a 
holonomic system in $D$-module theory. 
Hence combined with the integration algorithm for $D$-modules, 
this enables us to compute a holonomic system for 
the integral of a function involving the logarithm of a polynomial 
with respect to some variables. 
\end{abstract}

\section{Introduction}

For a given function $u$, it is an interesting problem both in 
theory and in practice to determine the differential equations 
which $u$ satisfies. Let us restrict our attention to 
linear differential equations with polynomial coefficients. 
Then our problem can be formulated as follows:
Let $D_n$ be the Weyl algebra, i.e., the ring of differential 
operators with polynomial coefficients in the variables
$x = (x_1,\dots,x_n)$. 
An element $P$ of $D_n$ is expressed as a finite sum
\begin{equation}\label{eq:anOperator}
 P = \sum_{\alpha,\beta \in \N^n} a_{\alpha,\beta}x^\alpha\partial^\beta,
\end{equation}
with $x^\alpha = x_1^{\alpha_1}\cdots x_n^{\alpha_n}$, 
$\partial^\beta = \partial_1^{\beta_1}\cdots\partial_n^{\beta_n}$ 
and $a_{\alpha,\beta}\in \C$, 
where
$\alpha = (\alpha_1,\dots,\alpha_n),\,\,\beta = (\beta_1,\dots,\beta_n)
\in \N^n$ are multi-indices with $\N = \{0,1,2,\dots\}$ 
and $\partial_i = \partial/\partial x_i$ ($i=1,\dots,n$) denote
derivations.  
The {\em annihilator} of $u$ (in $D_n$) is defined to be
\[
 \Ann_{D_n}u = \{P \in D_n \mid Pu=0\},
\]
which is a left ideal of $D_n$. 
Since $D_n$ is a non-commutative Noetherian ring, there exist a 
finite number of operators $P_1,\dots,P_N \in D_n$ which generate
$\Ann_{D_n}u$ as left ideal. Thus we can regard the system
\[
P_1 u = \cdots = P_Nu = 0
\]
of linear (partial or ordinary) differential equations as a maximal
one that $u$ satisfies. 

As to systems of linear differential equations, 
there is a notion of holonomicity, or being {\em holonomic},  
which plays a central role in $D$-module theory.
See Appendix for a precise definition. 
A holonomic system of linear differential equations admits only 
a finite number of linearly independent solutions although  
it is not a sufficient condition for holonomicity. 
A {\em holonomic function} is by definition a function which satisfies 
a holonomic system. 

The importance of the holonomicity lies in, in addition to the 
finiteness property above, the fact the it is preserved under 
basic operations on functions such as sum, product, restriction 
and integration. 
Hence starting from some basic holonomic functions we can construct
various holonomic functions by using such operations. 

As one of basic holonomic functions, 
let us consider $f^\lambda$ with a non-constant polynomial $f$ 
in $x = (x_1,\dots,x_n)$ and a complex number $\lambda$. 
Then the function $f^\lambda$ is holonomic and there is an algorithm 
to compute its annihilator strictly (\cite{OakuJPAA1997},\cite{BM},\cite{SST}). 

Our purpose is to give an algorithm to compute the annihilator of 
$f^\lambda (\log f)^m$ with a positive integer $m$ 
and to prove that it is a holonomic function.  
This is achieved by differentiation with respect to the parameter
$s$ of the annihilator of $f^s$ in $D_n[s]$. 
This method can be extended to functions of the form
$f_1^{\lambda_1}\cdots f_N^{\lambda_N}
 (\log f_1)^{m_1}\cdots(\log f_N)^{m_N}$ for polynomials 
$f_k$, complex numbers $\lambda_k$ and nonnegative integers $m_k$. 

Since the algorithm yields a holonomic system, we can apply 
the integration algorithm for $D$-modules 
(see \cite{OTdeRham}, \cite{SST})
to get a holonomic system for the integral of 
a function involving the logarithm of a polynomial.

\section{Annihilators with a parameter}

Let $f$ be a non-constant polynomial in $n$ variables 
$x = (x_1,\dots,x_n)$
with coefficients in the field $\C$ of the complex numbers. 
%In the sequel, the field $\C$ can be replaced by its arbitrary 
%subfield of. 
From an algorithmic viewpoint, we assume that 
the coefficients of $f$ belong to a computable field. 

First, we consider formal functions of the form $f^s(\log f)^k$ with an 
indeterminate $s$. More precisely, for a non-negative integer $m$, 
we introduce the module 
\[
\Lsc(f,m) := \bigoplus_{k=0}^m \C[x,f^{-1},s]f^s(\log f)^k,
\]
of which $f^s(\log f)^k$ are regarded as a free basis over $\C[x,f^{-1},s]$.
Then $\Lsc(f,m)$ has a natural structure of left $D_n[s]$-module, 
which is induced by the action
of the derivation $\partial_j = \partial/\partial x_j$ defined by, 
for $a \in \C[x,f^{-1},s]$,  
\[
\partial_j\{af^s(\log f)^k\} 
= \left(\frac{\partial a}{\partial x_j} + sa f^{-1}\frac{\partial f}{\partial x_j}\right)f^s(\log f)^k 
+ k a f^{-1}\frac{\partial f}{\partial x_j}f^s(\log f)^{k-1}
\quad (j=1,\dots,n)
\]
if $k \geq 1$ and 
\[
\partial_j(af^s) 
= \left(\frac{\partial a}{\partial x_j} + sa f^{-1}\frac{\partial f}{\partial x_j}\right)f^s
\quad (j=1,\dots,n).
\]
In view of this action, it is easy to see that $\Lsc(f,m)/\Lsc(f,m-1)$ 
is isomorphic to $\Lsc(f,0) = \C[x,f^{-1},s]f^s$ as a left $D_n[s]$-module. 

Now consider the left $D_n[s]$-submodule 
\[
 \Psc(f,m) := D_n[s]f^s + \cdots + D_n[s](f^s(\log f)^m)
\]
of $\Lsc(f,m)$.
Our purpose is to determine the annihilator module 
\[
\Ann_{D_n[s]}(f^s,\dots,f^s(\log f)^m)
:= \left\{P = (P_0,P_1,\dots,P_m) \in D_n[s]^{m+1}
\mid \sum_{k=0}^m P_k(f^s(\log f)^k) = 0 \right\}
\]
and the annihilator ideal
\[
\Ann_{D_n[s]}f^s(\log f)^m
:= \left\{P \in D_n[s]
\mid  P(f^s(\log f)^m) = 0 \right\}.
\]
Note that there are isomorphisms
\begin{align}
\Psc(f,m) &\simeq D_n[s]^{m+1}/\Ann_{D_n[s]}(f^s,\dots,f^s(\log f)^m), 
\nonumber\\
D_n[s](f^s(\log f)^m) &\simeq D_n[s]/\Ann_{D_n[s]}(f^s(\log f)^m).
\nonumber
\end{align}

Now let us regard $f^s\log f$ as a multi-valued analytic function 
in $(x,s)$ on $\{(x,s)\in \C^{n+1} \mid f(x) \neq 0 \}$.  

\begin{lemma}\label{lemma:independence}
Let $f \in \C[x]$ be a non-constant polynomial. 
Then for $a_i(x) \in \C[x]$, 
\[
 \sum_{i=0}^m a_i(x)(\log f)^i = 0
\]
holds as analytic function if and only if $a_i(x,s) = 0$ for 
all $i$. 
\end{lemma}

\begin{proof}
We argue by induction on $m$. 
Let $x_0 \in \C^n$ be a non-singular point of the hypersurface 
$f(x) = 0$, i.e, assume
\[
f(x_0) = 0,\quad \frac{\partial f}{\partial x_i}(x_0) \neq 0 
\quad \mbox{for some $i$ with $1 \leq i \leq n$}. 
\]
In view of the uniqueness of analytic continuation, 
we have only to show that each $a_i(x)$ vanishes near $x_0$. 
Hence we may suppose that $a_i(x)$ are analytic near $x_0 = 0$ 
and $f(x) = x_1$. That is, 
\begin{equation}\label{eq:identity} 
 a_0(x) + a_1(x)\log x_1 + \cdots + a_m(x)(\log x_1)^m = 0
\end{equation}
holds on a neighborhood $U$ of $0$. 
Fix a point $x = (x_1,\dots,x_n)$ in $U$ 
such that $x_1 \neq 0$. 
By analytic continuation along a circle 
$(e^{\sqrt{-1}t}x_1,x_2,\dots,x_n)$ with $0 \leq t \leq 2\pi$, 
the identity (\ref{eq:identity}) is transformed to 
\[
 a_0(x) + a_1(x)(\log x_1 + 2\pi\sqrt{-1}) + \cdots 
+ a_m(x)(\log x_1 + 2\pi\sqrt{-1})^m = 0.
\]
By subtraction, we get an identity of the form
\[
b_0(x) + b_1(x)\log x_1 + \cdots b_{m-1}(x)(\log x_1)^{m-1} = 0
\]
with
\[
b_{m-1}(x) = 2m\pi\sqrt{-1}a_m(x).
\]
From the induction hypothesis it follows that
$b_0(x) = \cdots = b_{m-1}(x)=0$, which implies $a_m(x)=0$. 
We are done by induction on $m$.
\end{proof}

\section{Computation of the annihilator}

Now let us describe an algorithm for computing the annihilator of
$f^s(\log f)^m$.  

\begin{algorithm}\label{alg:annfslog}
\rm
Input: a non-constant polynomial $f$ in the variables $x = (x_1,\dots,x_n)$ 
with coefficients in a computable subfield of $\C$, a non-negative 
integer $m$. 

\begin{enumerate}
\item
Let $G = \{P_1(s),\dots,P_k(s)\}$ be 
a generating set of the left ideal 
$\Ann_{D_n[s]}f^s := \{P(s) \in D_n[s] \mid P(s)f^s = 0\}$ 
by using an algorithm of \cite{OakuJPAA1997} or \cite{BM} 
(see also \cite{LM1}). 
\item
Let $e_0 = (1,0,\dots,0),\dots,e_{m} = (0,\dots,0,1)$ be the 
canonical unit vectors of $\C^{m+1}$. 
For each $i = 1,\dots,k$ and $j = 0,1,\dots,m$, set
\[
P_{i}(s)^{(j)} := \sum_{\nu = 0}^j { j \choose \nu}
\frac{\partial^{j-\nu}P_i(s)}{\partial s^{j-\nu}}e_\nu.
\]
\end{enumerate}

\noindent
Output: $G' := \{P_{i}(s)^{(j)} \mid 1\leq i \leq k,\,0\leq j \leq m\}$ 
generates
$\Ann_{D_n[s]}(f^s,\dots,f^s(\log f)^m)$. 
\end{algorithm}

\begin{algorithm}
\rm
Input: a non-constant polynomial $f$ in the variables $x = (x_1,\dots,x_n)$ 
with coefficients in a computable subfield of $\C$, a non-negative 
integer $m$. 

\begin{enumerate}
\item
Let $G'$ be the output of Algorithm \ref{alg:annfslog}. 
\item
Compute a Gr\"obner base $G''$ of the module generated by $G'$ 
with respect to a term order $\prec$ for $(D_n[s])^{m+1}$ such that
$Me_j \prec M'e_k$ for any monomial $M$ and $M'$ if $k < j$. 
Let $G_0$ be the set of the last component of each element of $G''$. 
\end{enumerate}
Output: $G_0$ generates $\Ann_{D_n[s]} f^s(\log f)^m$.
\end{algorithm}

\begin{lemma}\label{lemma:Ps}
Let $I$ be a left ideal of $D_n[s]$ generated by 
$\{P_1(s),\dots,P_k(s)\}$. For $P(s) \in D_n[s]$, and $j \in \N$, set
\[
P(s)^{(j)} := \sum_{\nu = 0}^j 
{j \choose \nu}\frac{\partial^{j-\nu}P(s)}{\partial s^{j-\nu}}e_\nu.
\]
Then the left $D_n[s]$-submodule of $(D_n[s])^{m+1}$ which is generated by 
$\{P(s)^{(j)} \mid P(s) \in I,\,0 \leq j \leq m\}$ 
coincides with the one which is generated by
$\{P_i(s)^{(j)} \mid 0 \leq i \leq k,\,0 \leq j \leq m\}$ 
for any integer $m \geq 0$. 
\end{lemma}

\begin{proof}
Let $\Nsc$ be the left $D_n[s]$-module generated by 
$\{P_i(s)^{(j)} \mid 0 \leq i \leq k,\,0 \leq j \leq m\}$
and $P(s)$ be a nonzero element of $I$.
Then there exist
\[
Q_i(s) = \sum_{l=0}^{m_i}Q_{il}s^l \quad (Q_{il} \in D_n)
\]
such that $P(s) = \sum_{i=1}^k Q_i(s)P_i(s)$. 
Then we have
\[
P(s)^{(j)} = \sum_{i=1}^k\sum_{l=0}^{m_i}Q_{il}(s^lP_i(s))^{(j)}.
\]
Hence we have only to show that $(s^lP_i(s))^{(j)}$ belongs to $\Nsc$. 
This can be done as follows:
\begin{align}
(s^lP_i(s))^{(j)}
&= \sum_{\nu=0}^j
{j \choose \nu}\left(\frac{\partial}{\partial s}\right)^{j-\nu}
(s^lP_i(s))e_\nu
\nonumber \\
&= \sum_{\nu=0}^j
{j \choose \nu}\sum_{\mu=0}^{\min\{j-\nu,l\}}
{j-\nu \choose \mu}(l)_\mu s^{l-\mu}
\left(\frac{\partial}{\partial s}\right)^{j-\nu-\mu}P_i(s)e_\nu
\nonumber\\
&= \sum_{\mu=0}^{\min\{j,l\}}{j \choose \mu} (l)_\mu s^{l-\mu}
\sum_{\nu=0}^{j-\mu}{j-\mu \choose \nu}
\left(\frac{\partial}{\partial s}\right)^{j-\mu-\nu}P_i(s)e_\nu
\nonumber\\
&= \sum_{\mu=0}^{\min\{j,l\}}{j \choose \mu} 
(l)_\mu s^{l-\mu}P_i(s)^{(j-\mu)},
\nonumber
\end{align}
where $(l)_\mu := l(l-1)\cdots(l-\mu+1)$. 
\end{proof}

\begin{theorem}
The output of Algorithm \ref{alg:annfslog} 
coincides with $\Ann_{D_n[s]}(f^s,\dots,f^s(\log f)^m)$. 
\end{theorem}

\begin{proof}
Let $P(s)$ belong to $\Ann_{D_n[s]}f^s$. 
Differentiating the equation $P(s)f^s = 0$ with respect to $s$, we get
\[ 
\sum_{\nu=0}^j 
{ j \choose \nu}\frac{\partial^{j-\nu}P_i(s)}{\partial s^{j-\nu}}
(f^s(\log f)^\nu) = 0
\]
for $0 \leq j \leq m$. 
This shows that each $P_{i}(s)^{(j)}$ annihilates 
$(f^s,\dots,f^s(\log f)^m)$. 

Set $\Msc := \Ann_{D_n[s]}(f^s,\dots,f^s(\log f)^m)$. 
Let $\Nsc$ be the left $D_n[s]$-module generated by 
the output $G'$ of Algorithm \ref{alg:annfslog}. 
The argument above shows that
$\Nsc$ is a left $D_n$-submodule of $\Msc$. 
Hence we have only to prove $\Nsc = \Msc$. 
For this purpose 
let $\Nsc_j$ be the left $D_n[s]$-module generated by 
$\{P_i(s)^{(\nu)} \mid 1 \leq i \leq k,\, 0 \leq \nu \leq j\}$
and set
\begin{align}
\Msc_j &:=
\{(Q_0,Q_1,\dots,Q_m) \in \Msc \mid
Q_\nu = 0 \mbox{ if } \nu > j\}.
\nonumber
\end{align}
Let $Q(s) = (Q_0(s),\dots,Q_j(s),0,\dots,0)$ be an element of $\Msc_j$. 
Then
\[
 \sum_{\nu=0}^j Q_\nu(s)(f^s(\log f)^\nu) = 0 
\]
holds.
In view of the action of $D_n[s]$ on $\Lsc(f,j)$ 
noted in Section 1, this implies $Q_j(s)f^s = 0$. 
Hence $Q_j(s)^{(j)}$ belongs to $\Nsc_j$ by Lemma \ref{lemma:Ps}. 
It is easy to see that $Q(s) - Q_j(s)^{(j)}$ belongs to 
$\Msc_{j-1}$. This means $\Msc_j = \Nsc_j + \Msc_{j-1}$ 
for $1 \leq j \leq m$. 
Then we can show that
$\Nsc_j = \Msc_j$ holds for $1 \leq j \leq m$ 
by induction on $m$ noting $\Nsc_0 = \Msc_0$. 
\end{proof}

\begin{remark}\rm
If $f$ is weighted homogeneous, i.e., 
if there exist rational numbers $w_i$ such that 
$\sum_{i=1}^n w_i x_i\partial_i(f) = f$,
then $D_n[s]f^s(\log f)^m$ is isomorphic to $\Lsc(f,m)$ 
as left $D_n[s]$-module.  That is, we have an isomorphism
\[
 D_n[s]/\Ann_{D_n[s]} f^s(\log f)^m \simeq 
(D_n[s])^{m+1}/\Ann_{D_n[s]}(f^s,\dots,f^s(\log f)^m). 
\]
of left $D_n[s]$-module
In fact, this follows from the relations 
\[
  \left(\sum_{i=1}^n w_i x_i\partial_i - s\right)(f^s(\log f)^k) 
= k f^s(\log f)^{k-1} \qquad (k \geq 1).
\]
\end{remark}

\section{Specialization of the parameter}

Let us fix a complex number $\lambda$.
(From an algorithmic view point, we assume $\lambda$ lies in a 
computable subfield of the field $\C$.) 
We set
\[
\Lsc(f,m,\lambda) := \bigoplus_{k=0}^m \C[x,f^{-1}]f^\lambda(\log f)^k,
\]
where $f^\lambda(\log f)^k$ are regarded as a free basis over $\C[x,f^{-1}]$.
Substituting $\lambda$ for $s$ gives $\Lsc(f,m,\lambda)$ a natural 
structure of left $D_n$-module. In fact, one has
\[
\partial_j\{af^\lambda(\log f)^k\} 
 \left(\frac{\partial a}{\partial x_j} + \lambda a f^{-1}\frac{\partial f}{\partial x_j}\right)f^\lambda(\log f)^k 
+ k a f^{-1}\frac{\partial f}{\partial x_j}f^\lambda(\log f)^{k-1}
\quad (j=1,\dots,n)
\]
for $k \geq 1$ and 
\[
\partial_j(af^\lambda) 
= \left(\frac{\partial a}{\partial x_j} + \lambda a f^{-1}\frac{\partial f}{\partial x_j}\right)f^\lambda
\quad (j=1,\dots,n)
\]
with $a \in \C[x,f^{-1}]$. 
This implies that $\Lsc(f,m,\lambda)/\Lsc(f,m-1,\lambda)$ is isomorphic to 
$\Lsc(f,0,\lambda) = D_nf^\lambda$ as a left $D_n$-module. 
It follows that $\Lsc(f,m,\lambda)$ is holonomic since so is $D_nf^\lambda$ 
as was proved by Bernstein \cite{Bernstein}. 

Set
\[
 \Psc(f,m,\lambda) := D_nf^\lambda + \cdots + D_n(f^\lambda(\log f)^m)
\]
We define the annihilators of 
$(f^s,\dots,f^\lambda(\log f)^m)$ 
and of $f^\lambda(\log f)^m$ to be
\begin{align}
\Ann_{D_n}(f^\lambda,\dots,f^\lambda(\log f)^k)
&:= \{P = (P_0,P_1,\dots,P_m) \in (D_n)^{m+1}
\mid \sum_{k=0}^m P_k(f^\lambda(\log f)^k) = 0 \},
\nonumber \\
\Ann_{D_n}f^\lambda(\log f)^m
&:= \{P \in D_n \mid  P(f^\lambda(\log f)^m) = 0 \}
\nonumber
\end{align}
respectively. 
Then we have isomorphisms
\begin{align}
\Psc(f,m,\lambda) &\simeq (D_n)^{m+1}
 /\Ann_{D_n}(f^\lambda,\dots,f^\lambda(\log f)^k)
\nonumber\\
D_n(f^\lambda(\log f)^m) &\simeq D_n/\Ann_{D_n}(f^\lambda(\log f)^m).
\nonumber
\end{align}

In the sequel, we need information on the integral roots of 
the \emph{Bernstein-Sato polynomial} or the \emph{$b$-function}\/
of $f$, which is, by definition, the monic polynomial 
$b_{f}(s)$ of the least degree such that a formal functional equation 
\begin{equation}\label{eq:fseq}
P(s)f^{s+1} = b_{f}(s)f^s
\end{equation}
holds with some $P(s) \in D_n[s]$. 
The existence of such a functional equations was proved by 
Bernstein \cite{Bernstein}. 
It was proved by Kashiwara \cite{Kashiwara} that 
the roots of $b_f(s)=0$ are negative rational numbers. 
An algorithm to compute $b_f(s)$ and an associated operator 
$P(s)$ was given in \cite{OakuDuke}. 
The following proposition generalizes a result of 
Kashiwara \cite[Proposition 6.2]{Kashiwara}:

\begin{theorem} \label{prop:specialization}
Let $b_f(s)$ be the Bernstein-Sato polynomial of $f$, i.e., a polynomial in $s$ of 
the least degree such that $P(s)f^{s+1} = b_f(s)f^s$ holds with a $P(s) \in D_n[s]$. 
Let $\lambda$ be a complex number such that $b_f(\lambda-\nu)
\neq 0$ for any positive integer $\nu$.  
Then we have
\begin{align}
\Ann_{D_n}(f^\lambda,\dots,f^\lambda(\log f)^k)
&= \{P(\lambda) \mid P(s) \in 
\Ann_{D_n[s]}(f^s,\dots,f^s(\log f)^k),
\nonumber \\
\Ann_{D_n}f^\lambda(\log f)^m
&= \{ P(\lambda) \mid P(s) \in \Ann_{D_n[s]}f^s(\log f)^m\}. 
\nonumber
\end{align}
\end{theorem}

\begin{proof}
We have only to show the first equality.
Assume that $\sum_{k=0}^m P_k (f^\lambda(\log f)^k) = 0$ holds with $P_k \in D_n$. 
Then there exist non-negative integer $l \geq 0$ 
and polynomials $a_k(x,s) \in \C[x,s]$ such that 
\[
  \sum_{k=0}^m P_k (f^s(\log f)^k) 
  = (s-\lambda)\sum_{k=0}^m a_k(x,s)f^{s-l}(\log f)^k. 
\]
By using the functional equation (\ref{eq:fseq}), 
we can find an operator $Q(s) \in D_n[s]$ such that
\[
  b_f(s-1)\cdots b_f(s-l)f^{s-l} = Q(s)f^s.
\]
In view of the action of $D_n[s]$ on $\Lsc(f,m)$, 
there exist $a'_k(x,s) \in \C[x,s]$ and a non-negative integers $l_1$ 
such that
\begin{align}
  b_f(s-1)\cdots b_f(s-l)f^{s-l}(\log f)^m &= Q(s)\{f^s(\log f)^m\}
%\nonumber\\ &
+ \sum_{k=0}^{m-1} a'_k(x,s)f^{s-l_1}(\log f)^k.
\nonumber
\end{align}
Proceeding inductively, we conclude that there exist a polynomial $b(s) 
\in \C[s]$ which is a product (possibly with multiplicities) of $b_f(s-j)$ 
with $j\geq 1$ and operators $\widetilde Q_k(s) \in D_n[s]$ such that
\[
b(s) \sum_{k=0}^m a_k(x,s)f^{s-l}(\log f)^k 
= \sum_{k=0}^m \widetilde Q_k(s)\{f^s(\log f)^k\}. 
\] 
Hence
\[
\widetilde P(s) := 
b(s)\sum_{k=0}^m P_ke_k - (s-\lambda)\sum_{k=0}^m \widetilde Q_k(s)e_k
\]
belongs to $\Ann_{D_n[s]}(f^s,\dots,f^s(\log f)^k)$ and 
$b(\lambda)\sum_{k=0}^m P_ke_k = \widetilde{P}(\lambda)$.  
This completes the proof since $b(\lambda) \neq 0$ by the assumption.
\end{proof}

If $b_f(\lambda-\nu) = 0$ for some positive integer $\nu$, 
then set $\nu_0 := \max\{\nu \in \Z \mid b_f(\lambda-\nu) = 0\}$ and 
$\lambda_0 := \lambda-\nu_0$.  
Then $\lambda_0$ satisfies the condition of Theorem \ref{prop:specialization}. 
Then for $(P_0,\dots,P_m) \in (D_n)^{m+1}$, we have
\begin{align}
&
(P_0,\dots,P_m) \in \Ann_{D_n}(f^\lambda,\dots,f^\lambda(\log f)^k)
\nonumber\\
&\Leftrightarrow\quad
 (P_0f^{\nu_0},\dots,P_mf^{\nu_0}) \in 
\Ann_{D_n}(f^{\lambda_0},\dots,f^{\lambda_0}(\log f)^k)
\nonumber\\
&\Leftrightarrow\quad
 (P_0,\dots,P_m) \in 
\Ann_{D_n}(f^{\lambda_0},\dots,f^{\lambda_0}(\log f)^k) : f^{\nu_0}.
\nonumber
\end{align}
The module quotient in the last line can be obtained by
computing the module intersection or else by syzygy computation. 
Now let us describe two algorithms for module quotient in general.  
First, let us define the componentwise product of two elements 
$P = (P_0,\dots,P_m)$ and $Q = (Q_0,\dots,Q_m)$ of $(D_n)^{m+1}$ 
to be $PQ := (P_0Q_0,\dots,P_mQ_m)$. 
Let $N$ be a left $D_n$-submodule of $(D_n)^{m+1}$ and $P$ be a 
nonzero element of $(D_n)^{m+1}$. 
Then the module quotient $N:P$ is defined to be
\[
N:P := \{Q \in (D_n)^{m+1} \mid QP \in N\}, 
\]
which is a left $D_n$-submodule of $(D_n)^{m+1}$. 

\begin{algorithm}\label{alg:quotient}\rm
Input: A set $G_1$ of generators of a left $D_n$-submodule $N$ of 
$(D_n)^{m+1}$  and a non-zero element $P = (P_0,P_1,\dots,P_m)$ 
of $(D_n)^{m+1}$.  
\begin{enumerate}
\item
Introducing a new variable $t$, compute a Gr\"obner base 
$G_2$ 
of the left $D_n[t]$-module of $(D_n[t])^{m+1}$ which is generated by
$\{(1-t)P_ke_k \mid 0 \leq k \leq m \} 
\cup \{tQ \mid Q \in G_1 \}$ with respect to a term 
order $\prec$ such that $x^\alpha\partial^\beta e_j  
\prec te_k$ for any $j,k \in \{0,1,\dots,m\}$ and $\alpha,\beta \in \N^n$. 
\item
$G_3 := G_2 \cap (D_n)^{m+1}$.
\item
$G_4 := \{ Q/P \mid Q \in G_3\}$, where $Q/P$ denotes the element in 
$(D_n)^{m+1}$ such that $(Q/P)P = Q$ in the sense of compnentwise 
product.  
\end{enumerate}
Output: $G_4$ generates the module quotient $N:P$. 
\end{algorithm}

In fact, we can show in the same way as in the commutative case that
$G_3$ generates the left module $N \cap (D_n)^{m+1}P$.  
In particular, for each 
$Q \in G_3$, there exists $Q' \in (D_n)^{m+1}$ such that 
$Q = Q'P$. Let us denote this $Q'$ by $Q/P$.  
Then $Q'$ belongs to the quotient module $N : P$. 
Conversely, if $Q'$ belongs to $N : P$, then $Q'P$ belongs to 
$N \cap (D_n)^{m+1}P$.  Hence $Q'$ belongs to the module generated 
by $G_4$.
The correctness of the following algorithm should be clear: 

\begin{algorithm}\rm
Input: A set $G_1 = \{Q_1,\dots,Q_k\}$ 
of generators of a left $D_n$-submodule $N$ of 
$(D_n)^{m+1}$  and a non-zero element $P=(P_0,\dots,P_m)$ of $(D_n)^{m+1}$.  
\begin{enumerate}
\item
Compute a set $G_2$ of generators of the syzygy module
\[
 \Ssc := \{(S_0,S_1,\dots,S_m,S_{m+1},\dots,S_{m+k}) \in (D_n)^{m+k+1} 
\mid \sum_{j=0}^m S_jP_je_j + \sum_{j=m+1}^{m+k}S_{m+j}Q_j = 0\}
\]
via a Groebner base of the module generated by $P_je_j$ ($0 \leq j \leq m$)
and $Q_j$ ($1 \leq j \leq k$).  
\item
Let $G_3$ be the set of the first $m+1$ components of the elements of $G_2$.
\end{enumerate}
Output: $G_3$ generates the module quotient $N:P$. 
\end{algorithm}

Summed up, the annihilators for $(f^{\lambda}((\log f)^k)_{0 \leq k \leq m}$
and $(f^{\lambda}(\log f)^m$ are computed as follows:

\begin{algorithm}\label{alg:specialization}\rm
Input: a non-constant polynomial $f$ in the variables $x = (x_1,\dots,x_n)$ 
with coefficients in a computable subfield of $\C$, 
a number $\lambda$ which belongs to a computable subfield of $\C$, 
a non-negative integer $m$. 
\begin{enumerate}
\item
Compute a set $G_1$ of generators of 
$\Ann_{D_n[s]}(f^s,\dots,f^s((\log f)^k)$ 
by Algorithm \ref{alg:annfslog}.
\item
Compute the (global) Bernstein-Sato polynomial $b_f(s)$ of $f$ 
by using one of the algorithms in \cite{OakuDuke}, \cite{OakuJPAA1997},  
\cite{BM} or their modifications. 
\item
Let $\nu_0$ be the largest positive integer $\nu$ such that 
$b_f(\lambda-\nu) = 0$ if there are any such $\nu$. 
If there are no positive integer $\nu$ such that $b_f(\lambda-\nu)=0$, 
then set $\nu_0=0$. 
\item
Set $\lambda_0 := \lambda - \nu_0$  and 
$G_2 := G_1|_{s=\lambda_0}$ 
(substitute $\lambda_0$ for $s$ in each element of $G_1$). 
\item
If $\nu_0 > 0$, then let $G_3$ be a set of generators of the 
module quotient $\langle G_2\rangle : f^{\nu_0} 
= \langle G_2\rangle : (f^{\nu_0},\dots,f^{\nu_0})$, 
where $\langle G_2\rangle$ denotes the left module generated by $G_2$. 
\item
If $\nu_0 = 0$, then set $G_3 := G_2$.
\item
Compute a Gr\"obner base $G_4$ of the module generated by $G_3$ 
with respect to a term order $\prec$ for $(D_n)^{m+1}$ such that
$Me_j \prec M'e_k$ for any monomial $M$ and $M'$ if $k < j$. 
Let $G_5$ be the set of the last component of each element of $G_4$. 
\end{enumerate}
Output: $G_3$ generates $\Ann_{D_n}(f^\lambda,\dots,f^{\lambda}(\log f)^m)$; 
$G_5$ generates $\Ann_{D_n}f^{\lambda}(\log f)^m$.
\end{algorithm}

\begin{remark}\rm
In step (3) of the algorithm above, we need only integer roots 
of the $b$-function. Hence one can employ 
a method described in \cite{LM2} to determine all the 
integer roots of the $b$-function efficiently
without computing the whole $b$-function. 
\end{remark}

\section{Implementation and examples}

We have implemented the algorithms  
in a computer algebra system Risa/Asir \cite{asir}, 
which is capable of Groebner base computation of modules over 
the ring of differential operators as well as over the ring of 
polynomials. 

\begin{example}\rm
(one dimensional case) 
Let $f$ be a square-free polynomial in one variable $x$ with 
complex coefficients. 
Since  $\Ann_{D_1[s]}f^s$ is generated by $f\partial -sf'$, 
the annihilator module 
$\Ann_{D_1[s]}(f^s,\dots,f^s(\log f)^m)$ is generated by 
$m+1$ elements
\[
(f\partial_x-sf',0,\cdots,0),\quad
(-f',f\partial_x-sf',0,\cdots,0),\quad\cdots,\quad
(0,\cdots,0,-mf',f\partial_x-sf')
\]
with $\partial_x = d/dx$ and $f' = \partial(f)$.

Since the Bernstein-Sato polynomial of $f$ is $s+1$, the substitution 
$s = \lambda$ gives generators of 
$\Ann_{D_1}(f^\lambda,\dots,f^\lambda(\log f)^m)$ if $\lambda 
\neq 0,1,2,\dots$.
In particular, $\Ann_{D_1}(f^{-1},\dots,f^{-1}(\log f)^m)$ is generated
by 
\[
(\partial_x f,0,\cdots,0),\quad
(-f',\partial_x f,0,\cdots,0),\quad\cdots,\quad
(0,\cdots,0,-mf',\partial_x f). 
\]
In view of Algorithm \ref{alg:specialization}, 
we can verify that 
$\Ann_{D_1}(1,\dots,(\log f)^m) = 
\Ann_{D_1}(f^{-1},\dots,f^{-1}(\log f)^m) : f$
is generated by 
\[
(\partial_x,0,\cdots,0),\quad
(-f',f\partial_x,0,\cdots,0),\quad\cdots,\quad
(0,\cdots,0,-mf',f\partial_x).
\]
Explicit generators of 
$\Ann_{D_1}(\log f)^m$ for $m \geq 1$ 
would be complicated: 
For example, if $f = x^3-x$ and $m=1$, 
Algorithm \ref{alg:specialization} gives generators
\[ 
\begin{array}{l}
(3x^5-4x^3+x)\partial_x^2+(3x^4+1)\partial_x,
\\
(x^3-x)\partial_x^3+(-3x^4+9x^2-2)\partial_x^2+(-3x^3+3x)\partial_x
\end{array}
\]
of $\Ann_{D_1}\log f$, which is not generated by a single element. 
\end{example}

\begin{example}\rm
Set $f = x^2y^2+z^2$ with $n=2$ and $(x_1,x_2,x_3) = (x,y,z)$, 
$\partial_x = \partial/\partial x$ and so on.  
First $\Ann_{D_3[s]}f^s$ is generated by 
\[
\begin{array}{l}
-x \partial_x+y \partial_y, 
\quad
y \partial_y+z \partial_z-2 s,
\\
z \partial_x-y^2 x \partial_z,
\quad
z \partial_y-y x^2 \partial_z,
\\
-z \partial_x^2+y^3 \partial_z \partial_y+y^2 \partial_z. 
\end{array}
\]
Since the Bernstein-Sato polynomial of $f$ is 
$b_f(s) = (s+1)^3(2s+3)$, 
the substitution $s=-1$ gives a set of generators of 
$\Ann_{D_3}f^{-1}\log f$.
Then by ideal quotient computation we get a set of generators
\[
\begin{array}{l}
-x \partial_x+y \partial_y,
\quad
-z \partial_x+y^2 x \partial_z,
\\
\partial_y^2 + x^2 \partial_z^2,
\quad
\partial_x^2 + y^2 \partial_z^2,
\\
-z \partial_y+y x^2 \partial_z,
\quad
\partial_y \partial_x^2-z y \partial_z^3,
\\
-\partial_y^2 \partial_x+z x \partial_z^3,
\quad
y\partial_y\partial_x + z\partial_z\partial_x
\\
y \partial_z \partial_y+z \partial_z^2+\partial_z,
\\
y \partial_y^2 +z \partial_z \partial_y + \partial_y,
\\
z \partial_y \partial_x+z y x \partial_z^2-y x \partial_z,
\\
\partial_y^2 \partial_x^2 + z^2 \partial_z^4 + 2 z \partial_z^3
\end{array}
\]
of $\Ann_{D_3}\log f$.  Let us consider the integral
\[
u(t) := \int_{\R^3} e^{-t(x^2+y^2+z^2)}\log(x^2y^2+z^2)\,dxdydz,
\]
which is well-defined for $t > 0$. 
Then $u(t)$ satisfies ordinary differential equations 
\[
 P_1 u(t) = P_2u(t)=0
\]
with
\begin{align}
P_1 &=t^3 \partial_t^5+(2 t^4+17 t^2) \partial_t^4+(32 t^3+80 t) \partial_t^3
\nonumber\\&
+(-4 t^4+144 t^2+100) \partial_t^2+(-28 t^3+192 t) \partial_t-36 t^2+48,
\nonumber\\
P_2&=
t^3 \partial_t^4+(3 t^4+14 t^2) \partial_t^3+(2 t^5+35 t^3+52 t) \partial_t^2
\nonumber\\&
+(14 t^4+102 t^2+48) \partial_t+18 t^3+66 t.
\nonumber
\end{align}
\end{example}

\section{(Appendix) Holonomic systems}

Let us present a precise definition of holonomicity. 
We define the total or the $(\one,\one)$-order of nonzero $P \in D_n$ to be 
\[
\ord_{(\one,\one)}(P) := \max\{
|\alpha|+|\beta| = 
\alpha_1 + \cdots +\alpha_n + \beta_1 + \cdots + \beta_n
\mid a_{\alpha,\beta} \neq 0\}.
\]
We set $\ord_w(0) := -\infty$.  
This induces the filtration
\[
 F_k(D_n) := \{ P \in D_n \mid \ord_{(\one,\one)}(P) \leq k\} 
\quad (k \in \Z)
\]
on the ring $D_n$. 
Let $M$ be a left $D_n$-module and $\{F_k(M)\}_{k \in \Z}$ be 
a good $(\one,\one)$-filtration. This means the following properties:
\begin{enumerate}
\item
every $F_k(M)$ is a finite dimensional vector space over $\C$;
\item
$F_k(M) \subset F_{k+1}(M) \quad \mbox{for all $k \in \Z$}$;
\item
$\displaystyle \bigcup_{k\in\Z}F_k(M) = M$;
\item
$F_i(D_n)F_k(M) \subset F_{i+k}(M) \quad 
\mbox{for all $i,k\in \Z$}$;
\item
there exists $k_1 \in \Z$ such that 
$F_k(M) = 0$ for $k \leq k_1$;
\item
there exists $k_2 \in \Z$ such that 
$F_i(D_n)F_k(M) = F_{i+k}(M)$ for $k \geq k_2$. 
\end{enumerate}
Then there exists a polynomial in $k$ such that 
$\dim_{\C}F_k(M) = p(k)$ for sufficiently large $k$. 
The degree of $p(k)$ does not depend on the choice of a 
good $(\one,\one)$-filtration of $M$ and is called the 
dimension of the module $M$, which we denote by $d(M)$. 
It was proved by Bernstein \cite{Bernstein} that $d(M) \geq n$ 
if $M \neq 0$. 
The following definition is due to Bernstein \cite{Bernstein}: 
\begin{definition}\rm
A finitely generated left $D_n$-module $M$ is called a {\em holonomic 
system} if $d(M) \leq n$. 
We also call a left ideal $I$ of $D_n$ to be 
a {\em holonomic ideal}, by abuse of terminology, if the left $D_n$-module 
$D_n/I$ is holonomic.  
\end{definition}
Note that $d(M) \leq n$ is equivalent to $d(M) = n$ or $M=0$ in view 
of the Bernstein inequality stated above. 
The dimension $d(M)$ can be computed as the degree of the Hilbert function
from a Gr\"obner base with respect to a term order which is compatible 
with the total degree.  

Holonomicity is preserved by operations such as sum, product, 
restriction to affine subvarieties, and integration with respect to 
some of the variables (cf.\ \cite{Bernstein}, \cite{Bjork}) 
and they are computable (see e.g., \cite{OakuIntegral}). 
Let $R_n := \C(x)\langle\partial_1,\dots,\partial_n\rangle$ be
the ring of differential operators with rational function coefficients. 
A $D_n$-module $M$ is said to be of finite rank and the dimension is
called the rank of $M$, if $R_nM$ is a finite dimensional vector space 
over $\C(x)$. 
A holonomic $D_n$-module $M$ is of finite rank but 
the converse is not true in general. 
Note that there is an algorithm for a given $D_n$-module $M$ of 
finite rank to construct a holonomic $D_n$-module 
$\widetilde M$ and a surjective $D_n$-homomorphism
of $M$ to $\widetilde M$ (\cite{OTW},\cite{Tsai}). 
If $M$ is a system of differential 
equations of finite rank for an analytic function $u$, 
then we have an isomorphism  
$\widetilde \ D_n/\Ann_{D_n}u$. 

\begin{example}\rm
Set $f = x^2y^2+z^2$ and consider the function $f^{-1}$. 
It is easy to see that the operators
\begin{align}
 f\partial_x + \frac{\partial f}{\partial x} 
&= (x^2y^2+z^2)\partial_x + 2xy^2,
\nonumber \\
 f\partial_y + \frac{\partial f}{\partial y} 
&= (x^2y^2+z^2)\partial_y + 2x^2y,
\nonumber \\
 f\partial_z + \frac{\partial f}{\partial z} 
&= (x^2y^2+z^2)\partial_z + 2z 
\nonumber 
\end{align}
annihilate $f^{-1}$. Let $J$ be the left ideal generated 
by these three operators, a `naive' annihilator. 
Then the Hilbert function of the $D_3/J$ is 
\[
\frac{1}{30} k^5+ \frac14 k^4 + \frac76 x^3 + \frac54 x^2 ,
+ \frac{43}{10}k
\] 
which means that the degree of $D_3/J$ is 5 and hence
$D_3/J$ is not holonomic although it is of rank one. 
The true annihilator $I$ of $f^{-1}$ is generated by 
\[
\begin{array}{l}
3 z^2 \partial_x^2-2 y^3 \partial_z \partial_y-2 y^2 \partial_z,
\\
3 z^2 \partial_y-2 y x^2 \partial_z,
\\
3 z^2 \partial_x-2 y^2 x \partial_z,
\\
3 y \partial_y+2 z \partial_z+6,
\\
-x \partial_x+y \partial_y
\end{array}
\]
and the Hilbert function of $D_3/I$ is 
\[
\frac73k^3-\frac32k^2+\frac{43}{6}k-1,
\]
which implies that $D_3/I$ is holonomic. 
The Hilbert function of $D_3/\Ann_{D_3}\log f$ is
\[  
2k^3+\frac32k^2+\frac52k-1. 
\]
\end{example}

\end{document}